\documentclass[11pt]{article}
\usepackage{graphicx} 

\usepackage{fullpage}
\usepackage{macros}
\usepackage{tikz}
\usepackage{bm}
\usepackage{thmtools,thm-restate}
\usetikzlibrary{positioning}
\usetikzlibrary{arrows.meta}
\usetikzlibrary{calc}
\usetikzlibrary{shapes}

\def\gapCVP{\mathrm{CVP}}

\def\gapSVP{\mathrm{SVP}}

\renewcommand{\SVP}{\gapSVP}
\renewcommand{\CVP}{\gapCVP}

\newcommand{\mvp}{\mathsf{mvp}}

\usepackage[draft,multiuser,inline,nomargin]{fixme}
\usepackage{booktabs,multirow}
\fxusetheme{color}
\FXRegisterAuthor{r}{er}{\color{red}Rajendra}
\FXRegisterAuthor{all}{eall}{Note}

\mathtoolsset{centercolon}

\newcommand{\basis}{\mathbf{B}}

\title{On Beating $2^n$ for the Closest Vector Problem}

\author{Amir Abboud\thanks{Weizmann Institute of Science and INSAIT, Sofia University ``St. Kliment Ohridski''. This work is part of the project CONJEXITY that has received funding from the European Research Council (ERC) under the European Union's Horizon Europe research and innovation programme (grant agreement No.~101078482). Supported by an Alon scholarship and a research grant from the Center for New Scientists at the Weizmann Institute of Science. Partially funded by the Ministry of Education and Science of Bulgaria's support for INSAIT, Sofia University ``St. Kliment Ohridski" as part of the Bulgarian National Roadmap for Research Infrastructure.} \\Weizmann Institute of Science \\amir.abboud@weizmann.ac.il \and Rajendra Kumar\thanks{Indian Institute of Technology Delhi. Part of this work was done while at Weizmann Institute of Science. Supported by Chandruka New Faculty Fellowship at IIT Delhi.} \\Indian Institute of Technology Delhi \\rajendra@cse.iitd.ac.in}
\date{\today}

\begin{document}

\maketitle
The Closest Vector Problem (CVP) is a computational problem in lattices that is central to modern cryptography.
The study of its fine-grained complexity has gained momentum in the last few years, partly due to the upcoming deployment of lattice-based cryptosystems in practice.
A main motivating question has been if there is a $(2-\varepsilon)^n$ time algorithm on lattices of rank $n$, or whether it can be ruled out by SETH.

Previous work came tantalizingly close to a negative answer by showing a $2^{(1-o(1))n}$ lower bound under SETH if the underlying distance metric is changed from the standard $\ell_2$ norm to other $\ell_p$ norms (specifically, any norm where $p$ is not an even integer).
Moreover, barriers toward proving such results for $\ell_2$ (and any even $p$) were established.

In this paper we show \emph{positive results} for a natural special case of the problem that has hitherto seemed just as hard, namely $(0,1)$-$\CVP$ where the lattice vectors are restricted to be sums of subsets of basis vectors (meaning that all coefficients are $0$ or $1$).
All previous hardness results applied to this problem, and none of the previous algorithmic techniques could benefit from it. 
We prove the following results, which follow from new reductions from $(0,1)$-$\CVP$ to weighted Max-SAT and minimum-weight $k$-Clique.

\begin{itemize}
\item An $O(1.7299^n)$ time algorithm for exact $(0,1)$-$\CVP_2$ in Euclidean norm, breaking the natural $2^n$ barrier, as long as the absolute value of all coordinates in the input vectors is $2^{o(n)}$.

\item A computational equivalence between $(0,1)$-$\CVP_p$ and Max-$p$-SAT for all even $p$ (a reduction from Max-$p$-SAT to $(0,1)$-$\CVP_p$ was previously known).

\item The minimum-weight-$k$-Clique conjecture from fine-grained complexity and its numerous consequences (which include the APSP conjecture) can now be supported by the hardness of a lattice problem, namely $(0,1)$-$\CVP_2$.

\end{itemize}

Similar results also hold for the Shortest Vector Problem.

\newpage

\maketitle

\section{Introduction}

\def \eps {\varepsilon}

A lattice $\lat$ of rank $n$ is the set of all integer linear combinations of a set of $n$ linearly independent \emph{basis} vectors $\basis=(\vec{b}_1,\vec{b}_2,\ldots,\vec{b}_n)\in \Q^{m\times n}$, \textit{i.e.}
\[\lat=\lat(\basis):=\left\{\sum_{i=1}^n z[i]\vec{b}_i\;|\; z[i]\in \Z\right\}.\]

The two most important computational problems on lattices are the Shortest Vector Problem (SVP) and the Closest Vector Problem ($\CVP$).
Given a basis $\basis$, $\SVP$ asks to find a shortest non-zero vector in the lattice $\lat(\basis)$, while in $\CVP$, we are also given a target vector $\vec{t}$ and want to find a closest vector in $\lat(\basis)$ to $\vec{t}$. For any approximation factor $\gamma\geq 1$, $\gamma$-$\SVP$ asks to find a non zero lattice vector whose length is atmost $\gamma$ times the length of shortest non-zero lattice vector. Similarly in $\gamma$-$\CVP$ we need to find a lattice vector whose distance from the target is at most $\gamma$ times the minimum distance between the lattice and the target vector.

Starting from the celebrated \textsf{LLL} algorithm by Lenstra, Lenstra, and Lov{\'a}sz in 1982~\cite{LLL82}, these problems have found various applications in algorithmic number theory~\cite{LLL82}, convex optimization~\cite{Kannan87,FrankT87}, cryptanalytic tools~\cite{Shamir84,Brickell84,LagariasO85}, and most importantly in modern cryptography where the security of many cryptosystems~\cite{Ajtai96,MR04,regevLatticesLearningErrors2009,Regev06,MR08,Gentry09,BV14} is based on their hardness.

The first reason cryptographers are excited about these problems is that they may be the key to the holy grail of basing security on NP-hard problems.
The motivation for this was in Ajtai's works that (1) proved the NP-hardness of $\SVP$ \cite{ajtaiShortestVectorProblem1998} (it was already known for $\CVP$ \cite{vEB81,arora1997hardness}),
and (2) designed a cryptographic hash function that is secure assuming $\SVP$ is hard to approximate up to a $\poly(n)$ factor \cite{Ajtai96}. 
Many follow-up works have tried to reduce the gap between the approximation factors that are provably NP-hard \cite{caiApproximatingSVPFactor1999,Mic01svp,khot2006hardness,khotHardnessApproximatingShortest2005, havivTensorbasedHardnessShortest2012,micciancioInapproximabilityShortestVector2012} and those on which crypto can be based \cite{MR04,regevLatticesLearningErrors2009,peikertPublickeyCryptosystemsWorstcase2009,BLPRS13}.
However, there is still a gap with certain formal barriers against closing it \cite{GG00,aharonovLatticeProblemsNP2005}.

The second reason for excitement is that these problems are believed to be hard for \emph{quantum} algorithms as well, and so lattice-based cryptosystems such as Regev's \cite{regevLatticesLearningErrors2009} are suitable for post-quantum cryptography. See \cite{peikertDecadeLatticeCryptography2016} for a survey.
Indeed, such a scheme~\cite{ABD+CRYSTALSKyberVersion022021,NIST2022} will soon replace the currently used number-theoretic schemes that are known to be breakable if a large quantum computer is built~\cite{shor94}.
In practice, where there is a trade-off between security and efficiency, system designers assume the hardness of approximate $\SVP$ or $\CVP$ in a precise, \emph{fine-grained} sense and work with the smallest possible instance sizes that are intractable.
Thus, even a ``mild'' improvement from $2^n$ to $2^{n/10}$ could break systems currently believed to be secure.

Many papers across the last 25 years aim to improve the base of the exponent for $\CVP$ and $\SVP$. There is an efficient reduction from $\SVP$ to $\CVP$ that preserves the rank (and the approximation factor) \cite{GMSS99} and so algorithms for $\CVP$ transfer to $\SVP$, but the other direction is not known.
The first exact algorithm was designed by Kannan~\cite{kannan1987algorithmic} and it had $n^{O(n)}$ time complexity for both problems. 
 Ajtai, Kumar and Sivakumar introduced a randomized sieving technique and proposed a $2^{O(n)}$ time algorithm for $\SVP$~\cite{AKS01}. They further extended this result to an approximation of $\CVP$~\cite{AKS02}. A sequence of works focused on improving the time complexity of this algorithm by optimizing the constant in the exponent~\cite{NV08,PS09, MV10, LWXZShortestLattice11}. Currently, the  fastest algorithm for $\SVP$ runs in $2^{n+o(n)}$ time~\cite{ADRS15,AS18}. (Furthermore, Aggarwal, Chen, Kumar, and Shen~\cite{ACKS20} recently demonstrated a quantum improvement for $\SVP$ with $1.784^{n}$ time.) 
For $\CVP$, Micciancio and Voulgaris~\cite{MV13} gave a deterministic algorithm that runs in $4^{n+o(n)}$ time. Aggarwal, Dadush and Stephens-Davidowitz~\cite{ADS15} gave the current fastest algorithm for $\CVP$, which has a time complexity of $2^{n+o(n)}$, matching the bound of $\SVP$. 
This remains the state of the art if we allow a $(1+\eps)$-approximation, but it can be improved if we allow larger factors  \cite{LWXZShortestLattice11,WLWFindingShortest15,EV20,ALS21}.

The search for a matching fine-grained $2^{(1-o(1))n}$ lower bound under popular assumptions such as SETH\footnote{The Strong Exponential Time Hypothesis (SETH) states that there is no $\eps>0$ such that for all $k\geq 3$ the $k$-SAT problem can be solved in $O((2-\eps)^n)$ time.} was kick-started by Bennet, Golovnev, and Stephens-Davidowitz \cite{BGS17}.
The authors were able to show a $2^{\Omega(n)}$ lower bound (under ETH) and a higher but seemingly sub-optimal lower bound of $2^{\omega n/3} \leq \Omega(1.17298^n)$ assuming that the current fastest algorithm of MAX-2-SAT is optimal. 
The original results were for $\CVP$ but later works extended them to $\SVP$ and other lattice problems as well~\cite{AS18b,BP20,AC21,ABGS19,BPT22} but with weaker bounds; e.g. there is currently no $1.0001^n$ lower bound for $\SVP$. 
Remarkably, the desired $2^{(1-o(1))n}$ lower bound under SETH was successfully accomplished for $\CVP$ \emph{if} we change the norm from Euclidean to $\ell_p$ where $p$ is anything but an even integer \cite{BGS17, ABGS19}; i.e. for the $\CVP_p$ problems where $p\not\in 2\Z_{>0}$.
Similar but weaker results hold for $\SVP_p, p\not\in 2\Z_{>0}$ as well \cite{AS18b}.
However, the Euclidean case (which is widely acknowledged to be by far the most popular) has remained tantalizingly open.

\begin{openproblem}
Can $\CVP$ and $\SVP$ (under the $\ell_2$ norm) be solved in $(2-\eps)^n$ time?
\end{openproblem}

The Euclidean case has been easier than other norms for the existing techniques. It is still unknown whether $2^{n+o(n)}$ time can be achieved for any other $\ell_p$ norm with $p\neq 2$:
the fastest algorithm in the exact case still requires $(\log n)^{\Omega(n)}$ time~\cite{reis2023subspace}, and faster constant-factor approximation algorithms are known if the ambient dimension $m$ is small enough~\cite{BN09,EV20}.
This could be because the other norms are more complicated and fewer people have thought about them, but it could also be that the Euclidean case is easier and that the $2^n$ lower bound for other norms is too pessimistic.

Interesting barrier results have been shown against the possibility of basing a $2^n$ lower bound for \emph{even norms} (i.e. $\ell_p$ with $p \in 2\Z_{>0}$) on SETH.
First, Aggarwal, Bennett, Golovnev, and Stephens-Davidowitz~\cite{ABGS19} showed that ``natural" reductions cannot show a lower bound for $\CVP$ higher than $2^{3n/4}$ under SETH.\footnote{A reduction is said to be natural if there is a bijective mapping between the set of satisfying assignments, and the set of closest vectors in the lattice.}
More recently, Aggrawal and Kumar showed that any $2^{\eps n}$ lower bound from SETH that is proved via a Turing reduction would collapse the polynomial hierarchy \cite{AK23}.
These results point out a technical difference between even and odd norms but it was unclear whether this difference could make the problems truly easier.

\subsection{Our Results}

In this paper we present new algorithms suggesting that, in the Euclidean norm, $\CVP$ and $\SVP$ are easier than previously thought, or rather that they are genuinely easier under assumptions that were hitherto considered mild. 

Specifically, our results concern the $(0,1)$-$\CVP$ and $(0,1)$-$\SVP$ variants where the solution must be a linear combination of the given basis vectors where each coefficient $z[i]$ is $0$ or $1$. 
This is not only a natural problem (reminiscent of Subset-Sum since each lattice vector is defined by a subset sum of basis vectors) but it is also the problem directly considered in all existing hardness results for the general problems. That is, the  known complexity theoretic results for $\CVP$ (or $\SVP$) are in fact hardness results for the special case of $(0,1)$-$\CVP$ (or $(0,1)$-$\SVP$).
Moreover, in the (non-fine-grained) poly-time regime this restriction is equivalent to the general case: There is a reduction from $\CVP$ on rank $n$ lattices to $(0,1)$-$\CVP$ on rank $n^3$ lattices.\footnote{This follows implicitly from \cite{AK23}. We can transform the basis and the target vector such that the coefficients of the closest lattice vector are bounded by $2^{n^2}$. Then we can construct the lattice with basis vectors $\forall i\in [n], j\in [n^2]$ $[(D\cdot 2^{j}\vec{b}_i)^T (\vec{e}_i)^T]^T$ where $\vec{e}_i$ is a vector where the $i^{\text{th}}$ coordinate is $1$ and the rest are zero. It is easy to argue that for a sufficiently large integer $D$ we get an almost approximation factor preserving reduction from $\CVP$ on rank $n$ lattices to $\{0,1\}$-$\CVP$ on rank $n^3$ lattices.} However, their fine-grained complexity could be different; in particular, it is easy to get a $2^n$ upper bound for $(0,1)$-$\CVP_p$ for all $p$ whereas it is open for $\CVP_p$.
To our knowledge, no existing algorithmic techniques could improve the state of the art in the Euclidean case, under the $(0,1)$ restriction.

\begin{definition}[$(0,1)$-$\CVP$]
\label{def:01CVP}
 For any $p\in [1,\infty]$,  $(0,1)$-$\CVP_p$ is defined as follows: Given a basis $\basis\in \Z^{m\times n}$ of lattice $\lat$,\footnote{Any lattice $\lat \subset Q^m$ can be scaled by sufficiently large integer $D$ to make it $D\lat \subseteq \Z^m$
} a target vector $\vec{t}\in \Z^m$, and a number $d>0$, the goal is to distinguish between:
\begin{itemize}
\item \textsf{YES} instances where $\exists \vec{z}\in \{0,1\}^n$ for which $\|\basis \vec{z}-\vec{t}\|_p\leq d$, and
\item \textsf{NO} instances where  
$\forall \vec{z}\in \{0,1\}^n$ the distance $\|\basis \vec{z}-\vec{t}\|_p >  d$ is large.\footnote{In the literature it is more common to define it such that in the \textsf{NO} case for all lattice vectors  the distance from target is more than $d$ i.e., $\forall \vec{z}\in \Z^n$, $\|\basis \vec{z}-\vec{t}\|_p>d$. Both these problems can be trivially reduced to each other  by a Karp reduction.}
\end{itemize}
\end{definition}

The $(0,1)$-$\SVP$ problem is defined analogously (\cref{def:01SVP}) and can also be reduced to $(0,1)$-$\CVP$ by a slight modification in the general case reduction~\cite{GMSS99}.
Note that when the $p$ subscript is omitted we are in the Euclidean case of $p=2$.

\paragraph{Main Result}
Our main result is an algorithm breaking the natural $2^n$ bound for the $(0,1)$ version of SVP and CVP, assuming that the coefficients of the basis vectors are not extremely large.
Ultimately, the algorithm is obtained by a reduction to the problem of detecting a triangle in a graph, and then exploiting fast matrix multiplication. Thus, our bounds depend on the exponent $\omega < 2.371866$ \cite{duan2022faster}.

\begin{theorem}
\label{thm:main}
There is an \emph{exact} algorithm for $(0,1)$-$\CVP$ and for $(0,1)$-$\SVP$ that runs in time $ 2^{\omega n/3 +o(n)} \leq  \tilde{\mathcal{O}}((1.7299)^n)$ if the coordinates of the basis and target vectors are bounded by $2^{o(n)}$.
\end{theorem}

Before our work, the only algorithms beating the natural $2^n$ bound for $\CVP$ (even under the $(0,1)$-restriction) were either a large constant factor approximation in $1.7435^n$ time \cite{LWXZShortestLattice11,WLWFindingShortest15,EV20} or a $2^{n/2}$ time $\sqrt{n}$-approximation algorithm \cite{ALS21}. Our algorithm is thus the fastest for any approximation factor below $\sqrt{n}$. 
Notably, this establishes a \emph{separation} between the Euclidean and the odd norms because such a bound for the odd norms (even if only for the $(0,1)$ with coefficients in $2^{o(n)}$) refutes SETH.

\paragraph{Equivalence with Max-$p$-SAT}
The upper bounds we achieve for $(0,1)$-$\CVP$ are similar to the state-of-the-art for the \emph{weighted Max-$2$-SAT} problem: given a $2$-CNF formula in which every clause has a weight, find an assignment that maximizes the total weight of satisfied clauses.
(Note that this is essentially the Max-Cut problem.)
Moreover, the technique for beating the natural $2^n$ bound due to Williams \cite{williams2005new} is also similar, namely by reduction to triangle detection.

Indeed, a reduction from Weighted-\textsf{Max}-$p$-$\SAT$ to $(0,1)$-$\CVP_p$ (\cref{thm:sat-to-cvp}) was shown by Bennet, Golovnev, and Stephens-Davidowitz \cite{BGS17} to prove the hardness of $(0,1)$-$\CVP$.
Our second result is a reduction in the reverse direction, establishing a fine-grained \emph{equivalence}.

\begin{restatable}{theorem}{ThmCVPtoSAT}\label{thm:cvp-to-sat}
    For any $p\in 2\Z_{>0}$, there exists a poly-time many-one (Karp) reduction from $(0,1)$-$\CVP_p$  on lattices of rank $n$ to Weighted-\textsf{Max}-$p$-$\SAT$ on $n$ variables.  
\end{restatable}

\begin{corollary}
\label{cor:equiv}
For any $p\in 2\Z_{>0}$ and $T(n)$, the $(0,1)$-$\CVP_p$ problem can be solved in $O(T(n))+ n^{O(1)}$ time \emph{if and only if} Weighted-\textsf{Max}-$p$-$\SAT$ can. 
\end{corollary}

While the main result in \cref{thm:main} implicitly follows by combining the reduction in \cref{thm:cvp-to-sat} with the known Max-$2$-SAT algorithm of \cite{williams2005new}, we believe it can be more enlightening to see a more direct proof that does not go via Max-$2$-SAT.

A corollary of our results is that (in the $(0,1)$ case) $\CVP$ in even norms \emph{reduces} to $\CVP$ in odd norms (via a Karp reduction).
\begin{corollary}
    For any $p\in 2\Z_{>0}$ and $q\not\in 2\Z$, there exists a poly time many-one (Karp) reduction from $(0,1)$-$\CVP_p$ on lattice of rank $n$ to $(0,1)$-$\CVP_q$ on lattice of rank $n$. 
\end{corollary}

This follows because Max-$p$-SAT for any $p$ can be reduced to $\CVP_q$ for any $q$ except even integers.
The reverse direction (from any odd $q$ to any even $p$) would collapse the polynomial hierarchy \cite{AK23}. 
In a sense, we generalize the result of Regev and Rosen \cite{RR06} who reduced the $p=2$ case to any $\ell_q$ norm; however, the blowup in the approximation factor in our reduction is worse.

\paragraph{Connection to $k$-Clique}

Our algorithm is not only a reduction to triangle detection but also to the $k$-Clique problem for any $k\geq 3$. This is formally stated in~\cref{thm:CVPtoclique}. 
It may appear that our result is only a small step away from beating the $2^n$ bound for $(0,1)$-$\CVP$ on an arbitrary basis (without a bound on the coordinate values).
However, our technique is unlikely to achieve that without further breakthroughs, because it requires us to break the $n^k$ bound for min-weight-$k$-clique on arbitrary edge weights - refuting a conjecture in fine-grained complexity \cite{AbboudWW14,BackursDT16,BackursT17,BringmannGMW20}.
This is one of the most important conjectures in the field because it unifies two of the main three conjectures \cite{AbboudBDN18}, namely the All Pairs Shortest Paths (APSP) conjecture and the Orthogonal Vectors (OV) conjecture (which is the representative of SETH inside P). 
Viewed negatively, our result shows that breaking the $n^k$ bound for min-weight-$k$-clique is even harder than previously thought because it would be a breakthrough that all the cryptographers working on lattice problems have missed. We think this is interesting support for this conjecture and its numerous consequences (which include all APSP-based lower bounds, almost all SETH-based lower bounds, and several others; see \cite{williams2018some}).
Moreover, our reductions show that breaking the $n^{\omega k/3}$ bound of \emph{unweighted} $k$-Clique would improve the above results; thus also basing the unweighted $k$-Clique conjecture \cite{AbboudBW18} on the hardness of lattice problems. 

In \cref{thm:CVP-even-norm-to-clique} we generalize the reduction for any even $p \in 2\Z_{>0}$ to reduce $(0,1)$-$\CVP_p$ to $k$-Clique on $p$-uniform hyper-graphs. Unfortunately, the latter problem is unlikely to have non-trivial algorithms \cite{AbboudBDN18,LincolnWW18} when $p\geq 3$.

\subsection{Technical Overview}
\label{sec:overview}

The first idea in our reduction from $(0,1)$-$\CVP$ to minimum-weight $k$-Clique (\cref{thm:CVPtoclique}) is to use a \emph{split-and-list} approach; a standard technique in exponential time algorithms similar to the famous meet-in-the-middle algorithm for \emph{Subset Sum}.
We partition the $n$ basis vectors $\basis$ arbitrarily into $k$ sets $\basis^{(1)},\ldots, \basis^{(k)}$ of size $n/k$ each, and then for each set we enumerate all $N:=2^{n/k}$ vectors attainable by taking the sum of a subset of vectors from the set.
This produces $k$ lists $\mathbf{C}^{(1)},\ldots,\mathbf{C}^{(k)} $ of $N=2^{n/k}$ vectors each.
Observe that the closest vector we are looking for $\vec{v}=\sum_{i=1}^n z[i]\vec{b}_i$ can be represented as the sum of $k$ vectors $z = \vec{c}_1 +\cdots + \vec{c}_k$, one from each list $\vec{c}_i \in \mathbf{C}^{(i)}$, and moreover, any sum of $k$ vectors from $\mathbf{C}^{(1)} \times \cdots \times \mathbf{C}^{(k)}$ is a valid candidate for being the closest vector.
Thus, our task becomes to find the optimal way of picking one vector from each list.
A similar idea of applying split-and-list to the $(0,1)$-$\CVP$ problem (over any norm) was used by Gupte and Vaikuntanathan \cite{gupte2021fine} to prove conditional lower bounds for the \emph{Sparse Linear Regression} problem.

Superficially, it may seem that we are done: simply represent each vector $\vec{c}_i$ with a node and let a $k$-clique represent the sum of $k$ vectors. All we have to do is ensure that the total weight of the $k$-clique corresponds to the distance to the target $\vec{t}$. Then, the minimum-weight $k$-clique will give us the closest vector.

However, the total weight of a $k$-clique can only be influenced by \emph{pairwise} contributions from its $k$ nodes.
In the $\CVP$ interpretation, we require that the distance of the vector $z = \vec{c}_1 +\cdots + \vec{c}_k$ from the target $\vec{t}$ can be measured by only considering the sum of ${k \choose 2}$ values that depend only on $\vec{c}_{i}, \vec{c}_j$ for all $i,j \in [k]$, i.e. the sum of some pairwise contributions $\sum_{i,j \in [k]} f(\vec{c}_{i}, \vec{c}_j)$.
Unfortunately, this is impossible (and would refute SETH) \emph{unless we restrict the metric space}.

Our second (and main) idea is that \emph{under the Euclidean norm} the expression $\| \vec{c}_1 +\cdots + \vec{c}_k - t \|^2$ can be broken into a sum that depends only on pairs $\vec{c}_{i}, \vec{c}_j$ and can therefore be implemented as the weight of a $k$-clique under a careful choice of weights.
Interestingly, this ability to separate the distance in $\ell_p$ norm into $p$-wise contributions only holds when $p$ is even, and it is also the underlying technical reason that enables the barrier results of Aggarwal and Kumar \cite{AK23}. 

At a high level, the reduction to Max-SAT (\cref{thm:cvp-to-sat}) can be viewed as setting $k$ to be $n$ in the above reduction, and then adjusting several implementation details that have to do with SAT vs. Clique. 
When $k=n$ we are essentially partitioning the set of basis vectors into singletons and thinking of all possible $2^1=2$ choices of either choosing the vector or not. Naturally, each such singleton can be represented with a Boolean variable that determines if the corresponding vector is chosen in the solution.
Then it remains to encode the distance of the solution from the target using pairwise contributions (that can be encoded as the weight of a width-$2$ clause). Extra challenges (compared to the $k$-clique reduction) arise because we are forced to consider contributions from variables set to $0$ (because they may also satisfy clauses).

\section{Preliminaries}
\label{sec:prelim}

We will use $\R$, $\Z$, $\Z_{>0}$, and $\Q$ to represent the sets of real numbers, integers, positive integers, and rational numbers, respectively. For any positive integer $k$, we use $[k]$ to denote the set $\{1,2,\ldots,k\}$. We will use boldface lower-case letters to denote column vectors, e.g., $\vec{v}\in \R^m$, and we will use $v[i]$ to denote the $i^{\text{th}}$ coordinate of $\vec{v}$. We use boldface upper-case letters to denote a matrix, e.g., $\mat{M}\in \R^{m\times n}$ and $\vec{m}_i$ to denote the $i^{\text{th}}$ column vector of $\mat{M}$. For  vector $\vec{v}\in \R^m$, the $\ell_p$ norm of vector $\vec{v}$ for $p\in [1,\infty)$,  is defined as: 

\[\|\vec{v}\|_p:=\left(\sum_{i=1}^m |v[i]|^p\right)^{1/p},\]
and for $p=\infty$ it is defined as: 
\[\|\vec{v}\|_\infty:=\max_{i=1}^m\left\{ |v[i]|\right\}.\]
We omit the parameter $p$ when $p=2$ and write $\|\vec{v}\|$ to denote $\|\vec{v}\|_2$, a.k.a the Euclidean norm.

\subsection{Lattice Problems}
For any set of linearly independent vectors $\basis=(\vec{b}_1,\vec{b}_2,\ldots,\vec{b}_n)\in \Q^{m\times n}$ and for any positive integers $n$ and $m\geq n$, the lattice $\lat$ generated by Basis $\basis$ is defined as follows: 
\[\lat=\lat(\basis):=\left\{\sum_{i=1}^n z[i]\vec{b}_i\;|\; z[i]\in \Z\right\}.\]
Here, we call $n$ the rank of the lattice and $m$ the dimension. Note that a lattice can have infinitely many bases. We write $\lat(\basis)$ to emphasize that the lattice $\lat$ is generated by the basis $\basis$. 
For any vector $\vec{t}\in \R^m$, we use $\dist_p(\lat,\vec{t})$ to denote the distance of the vector $\vec{t}$ from the lattice $\lat$ in $\ell_p$ norm, i.e. $\dist_p(\lat, \vec{t})=\min_{\vec{v}\in \lat}\{\|\vec{v}-\vec{t}\|_p\}$. We will also use $\dist_p(\basis,\vec{t})$ in place of $\dist_p(\lat(\basis),\vec{t})$.

In this work, we focus on restricted versions of $\SVP$ and $\CVP$, namely $(0,1)$-$\CVP$ and $(0,1)$-$\SVP$.
Recall \cref{def:01CVP} ($(0,1)$-$\CVP$) from the introduction. 

\begin{definition}[$(0,1)$-$\SVP$]
\label{def:01SVP}
For any  $p\in [1,\infty]$, the $(0,1)$-$\SVP_p$ is a problem defined as follows: Given a basis $\basis\in \Z^{m\times n}$ of lattice $\lat$ and a number $d>0$, the goal is to distinguish between:
\begin{itemize}
\item \textsf{YES} instances where $\exists \vec{z}\in \{0,1\}^n\setminus \vec{0}$ for which $\|\basis\vec{z}\|_p\leq d$, and
\item \textsf{NO} instances where  
$\forall \vec{z}\in \{0,1\}^n \setminus \vec{0}$, $\|\basis \vec{z}\|_p >  d$.
\end{itemize}
\end{definition}

We  omit the parameter $p$ when $p=2$ and write $\CVP$, $\SVP$, $(0,1)$-$\CVP$, $(0,1)$-$\SVP$ for $\CVP_p$, $\SVP_p$, $(0,1)$-$\CVP_p$, $(0,1)$-$\SVP_p$ respectively. 

By a small modification to a reduction from \cite{GMSS99}, we can get a reduction from $(0,1)$-$\SVP$ on lattice of rank $n$ to $n$ instances of $(0,1)$-$\CVP$ on lattice of rank $n-1$. In this paper, we will only present our reduction/algorithm for $(0,1)$-$\CVP$. By above reduction, similar consequences will also hold for $(0,1)$-$\SVP$.  

Here, we have defined the decision versions of lattice problems, but our reductions also apply to the search versions of these problems. However, it is important to note that currently, we only have known search-to-decision reductions for lattice problems with an approximation factor slightly greater than one~\cite{SteSearchtodecisionReductions16}. It remains an open problem to show search-to-decision reductions for lattice problems with constant or larger approximation factors. 

\subsection{Satisfiability and Clique}
A $k$-$\SAT$ formula $\Psi=\wedge_{i=1}^m C_i $ on Boolean variables $x_1,\ldots, x_n$ is a conjunction of $m$ clauses $C_1,\ldots, C_m$ where each clause $C_i$ is a disjunction of at most $k$ literals and a literal is either a variable $x_j$ or its negation $\neg x_j$. 

\begin{definition}[Max-$k$-$\SAT$]
    Given a $k$-$\SAT$ formula $\Psi$ on $n$ variables and a number $\delta\in [0,1]$, the goal is to distinguish between \textsf{YES} instances where there exists an assignment that satisfies at least $\delta$ fraction of clauses of $\Psi$ and \textsf{NO} instances where all assignments satisfy a less than $\delta$ fraction of the clauses.
\end{definition}

\begin{definition}[Weighted Max-$k$-$\SAT$]
    Given a $k$-$\SAT$ formula $\Psi=\wedge_{i=1}^m C_i$ on $n$ variables and $m$ clauses $C=\{C_1,\ldots, C_m\}$, a clause weight function $w: C\rightarrow \Z_{>0}$, and a number $d$, the goal is to distinguish between \textsf{YES} instances where there exists an assignment for which the sum of weights of satisfied clauses is at least $d$ and \textsf{NO} instances where for all assignments the sum of weight of satisfied clauses is less than $d$.
\end{definition}

A hypergraph $\mathcal{G}=(\mathcal{{V},\mathcal{E}})$ consists of a finite set of vertices $\mathcal{V}$ and a set of hyperedges $\mathcal{E}$. In this work, we will only focus on $p$-uniform hypergraphs where all hyperedges are of exactly $p$ vertices. 
A $k$-clique is a set of $k$ nodes that have all ${k \choose p}$ hyper-edges between them, and its \emph{total weight} is defined as the sum of the weights of all of these edges.

\begin{definition}[minimum-weight-$k$-Clique]
For any $p\in \Z_{\geq 2}$, the minimum-weight-$k$-Clique problem defined as follows: Given a $p$-uniform hypergraph $\mathcal{G}=(\mathcal{{V},\mathcal{E}})$, a weight function $w: \mathcal{E}\rightarrow \Z_{>0}$, and an integer $d$, the goal is to distinguish between \textsf{YES} instances where there exists a Clique of $k$ vertices with weight at most $d$, and \textsf{NO} instances where all $k$-cliques have total weight greater than $d$.
\end{definition}

\subsection{Multi-Vector Products}
The following notion of multi-vector products ($\mvp$) defined in \cite{AK23} will be convenient for us. For any $p\in 2\Z_{>0}$, 
\[\forall \vec{v}_1,\ldots, \vec{v}_p\in \R^m: \mvp(\vec{v}_{1},\ldots,\vec{v}_p):=\sum_{i=1}^m \left(\prod_{j=1}^{p}v_j[i] \right).\]

Notice that it is an extension of inner-product to $\ell_p$ norm for any even integer $p$ since for any $\vec{v}_1,\vec{v}_2\in \R^m$, $\mvp(\vec{v}_1,\vec{v}_2)=\sum_{i=1}^m v_1[i]\cdot v_2[i]=\langle \vec{v}_1,\vec{v}_2 \rangle$.

Simple calculations prove the following helpful lemma stated in \cite[Lemma~4.1]{AK23}.

\begin{lemma}\cite[Lemma~3.1]{AK23}\label{lem:mvp}
	For any $p\in2\Z_{>0}$, and vectors $\vec{v}_1,\ldots,\vec{v}_k$, for any $a_1,\ldots, a_k\in \Z$, $\|a_1\vec{v}_1+\ldots +a_n\vec{v}_n\|_p^p$ can be computed in polynomial time given only $a_1,\ldots,a_k$ and $\mvp(\vec{v}_{i_1},\ldots, \vec{v}_{i_p})$ for all $(i_1,\ldots, i_p)\in [k]^p$. Moreover, 
	\[ \|a_1\vec{v}_1+\cdots+a_k\vec{v}_k\|_p^p=\sum_{(i_1,\ldots,i_p)\in [k]^p} (a_{i_1}\cdots a_{i_p})\mvp(\vec{v}_{i_1},\ldots, \vec{v}_{i_p}).
	\]
\end{lemma}

We will be using the following corollary of \cref{lem:mvp}.
\begin{corollary}\label{cor:simple-mvp}
	For any $p\in2\Z_{>0}$, and vectors $\vec{v}_1,\ldots,\vec{v}_k$, we have
	\[ \|\vec{v}_1+\cdots+\vec{v}_{k-1}-\vec{v}_{k}\|_p^p=\sum_{(i_1,\ldots,i_p)\in [k]^p} (-1)^{\sigma(i_1,\cdots,i_p)}\mvp(\vec{v}_{i_1},\ldots, \vec{v}_{i_p}).
	\]
 Here, for any tuple of $p$ elements $\mathcal{S}\subset [k]^p$, $\sigma(\mathcal{S})$ denote the number of occurrences of $k$ in $\mathcal{S}$.
\end{corollary}

\section{From CVP to Clique, and Algorithmic Consequences}
\label{sec:cvp-to-clique}

In this section, we present a reduction from $\{0,1\}$-$\CVP_2$ on lattice of rank $n$ to minimum-weight-$k$-Clique on undirected graph with $k\cdot 2^{\ceil{n/k}}$ vertices that is overviewed in \cref{sec:overview} and we will prove the main algorithmic results of this paper.

We will use the following well-known lemma\footnote{It can also be seen as a corollary of \cref{lem:mvp}.} about Euclidean norm.
\begin{lemma}\label{cor:euclidean-norm}
	For any vectors $\vec{v}_1,\ldots, \vec{v}_k,\vec{t}$,
	\[\|\vec{v}_1+\cdots+\vec{v}_k-\vec{t}\|^2=\|\vec{t}\|^2+\sum_{i=1}^{k} \left(\|\vec{v}_i\|^2-2\cdot  \langle \vec{v}_i,\vec{t} \rangle\right)+2\sum_{\substack{(i,j)\in [k]^2 \\\text{and } i<j}} \langle \vec{v}_i,\vec{v}_j \rangle . \] 
\end{lemma}

\begin{theorem}
\label{thm:CVPtoclique}
  For positive integers $n$, and $k\geq 2$,  there exists a Karp-reduction from $\{0,1\}$-$\CVP_2$ on lattices of rank $n$ to minimum-weight-$k$-Clique on undirected graphs with $N=k\cdot 2^{\lceil n/k\rceil}$ vertices.\\
  Furthermore, if the absolute value of the coordinates in all basis and target vectors is at most $2^\eta$, then the reduction takes only $O(m \cdot \eta^2\cdot (k\cdot 2^{\lceil n/k\rceil})^2)$ time and space. Additionally, the edge weights of the reduced graph are bounded by ${O}(m\cdot 2^{2\eta})$.
\end{theorem}

\begin{proof}
 Given a $(0,1)$-$\CVP$ instance with
 basis $\basis=(\vec{b}_1,\ldots \vec{b}_n)\in \Z^{m\times n} $, target $\vec{t}\in \Z^m$ and a number $d>0$, where each coordinate is bounded by $2^{\eta}$, we construct a minimum-weight-$k$-Clique instance on a graph $\mathcal{G}=(\mathcal{V},\mathcal{E})$ with $k\cdot 2^{n/k}$ vertices. For simplicity, let's assume that $n$ is an integer multiple of $k$.

    First, partition the basis vectors into $k$ sets $\basis^{(1)},\ldots, \basis^{(k)}$ each consisting of $n/k$ vectors such that $\basis^{(i)}=\left\{\vec{b}_{\frac{n}{k}(i-1)+1},\ldots, \vec{b}_{i\frac{n}{k}}\right\}$. We construct $\mathbf{C}^{(i)}=[\vec{c}_{1}^{(i)}\ldots \vec{c}_{2^{n/k}}^{(i)} ]\in \Z^{m\times 2^{n/k}}$ for $i\in [k]$, where the columns of $\mathbf{C}^{(i)}$ represents  $\{0,1\}$-combinations of vectors from $\mathbf{B}^{(i)}$ and each vector in $\mathbf{C}^{(i)}$ will correspond to a vertex in the reduced graph.
    Let
     $\mathcal{V}:=\{v_{j}^{(i)}| i\in [k], j\in [2^{n/k}]\}$ 
     be the set of vertices and 
     \[\mathcal{E}:=\{e_{i_1,j_1, i_2,j_2}=(v_{j_1}^{(i_1)},v_{j_2}^{(i_2)}) | i_1, i_2\in [k], i_1<i_2, j_1,j_{2}\in [2^{n/k}] \}\]
     be the edge set. Notice that it is a $k$-partite graph.  For all edge $e_{i_1,j_1, i_2,j_2}\in \mathcal{E}$, we define the edge weight as follows: 
    
       \[ w(e_{i_1,j_1,i_2,j_2}):= 2\langle\vec{c}_{j_1}^{(i_1)},\vec{c}_{j_2}^{(i_2)}\rangle+\frac{1}{k-1}\left(\|\vec{c}_{j_1}^{(i_1)}\|^2+\|\vec{c}_{j_2}^{(i_2)}\|^2 - 2\langle \vec{c}_{j_1}^{(i_1)},\vec{t} \rangle-2\langle \vec{c}_{j_2}^{(i_2)},\vec{t}\rangle\right)+\left(\binom{k}{2}\right)^{-1}\|\vec{t}\|^2 .\]
    
    It is trivial that the reduction takes $O(m\cdot \eta^2 \cdot (k\cdot 2^{n/k})^2)$ time and space. Notice that the edge weights are also bounded by $O(m\cdot2^{2\eta})$. In the rest of the proof, we will show that there exists a $k$-clique of weight at most $d^2$ if the given $(0,1)$-$\CVP_2$ instance is a \textsf{YES} instance, and otherwise all $k$-cliques have weight greater than $d^2$.
    
    We claim that the weight of a $k$-clique in the reduced graph is at least $\min_{\vec{z}\in \{0,1\}^n} \|\basis \vec{z}-\vec{t}\|^2$. Notice that, any $k$-clique in a $k$-partite graph must have exactly one vertex from each partition.  Let the vertices ${v}_{j_1}^{(1)},\ldots, {v}_{j_k}^{(k)}$ form a $k$-Clique, then the weight of the the clique is 
    \begin{align*}
    &\sum_{\substack{(x_1,x_2)\in [k]^2\\ \text{ and }x_1<x_2 }} w(e_{{x_1},j_{x_1}}, e_{{x_2},j_{x_2}})\\
    =& \sum_{\substack{(x_1,x_2)\in [k]^2\\ \text{ and }x_1<x_2 }}2\langle\vec{c}_{j_{x_1}}^{({x_1})},\vec{c}_{j_{x_2}}^{({x_2})}\rangle+\frac{1}{k-1}\left(\|\vec{c}_{j_{x_1}}^{({x_1})}\|^2+\|\vec{c}_{j_{x_2}}^{({x_2})}\|^2 - 2\langle \vec{c}_{j_{x_1}}^{({x_1})},\vec{t} \rangle-2\langle \vec{c}_{j_{x_2}}^{({x_2})},\vec{t}\rangle\right)+\left(\binom{k}{2}\right)^{-1}\|\vec{t}\|^2\\
     =& \|\vec{t}\|^2+ \sum_{\substack{(x_1,x_2)\in [k]^2\\ \text{ and }x_1<x_2 }}2\langle\vec{c}_{j_{x_1}}^{({x_1})},\vec{c}_{j_{x_2}}^{({x_2})}\rangle+\frac{1}{k-1}\left(\|\vec{c}_{j_{x_1}}^{({x_1})}\|^2+\|\vec{c}_{j_{x_2}}^{{x_2}}\|^2 - 2\langle \vec{c}_{j_{x_1}}^{{x_1}},\vec{t} \rangle-2\langle \vec{c}_{j_{x_2}}^{({x_2})},\vec{t}\rangle\right)\\
    =& \|\vec{t}\|^2+ \sum_{x\in[k]}\left(\|\vec{c}_{j_{x}}^{({x})}\|^2 - 2\langle \vec{c}_{j_{x}}^{({x})},\vec{t} \rangle\right)+  \sum_{\substack{(x_1,x_2)\in [k]^2\\ \text{ and }x_1<x_2 }}2\langle\vec{c}_{j_{x_1}}^{({x_1})},\vec{c}_{j_{x_2}}^{({x_2})}\rangle\\
    =& \|\vec{c}_{j_1}^{(1)}+\cdots + \vec{c}_{j_k}^{(k)}-\vec{t}\|^2 \geq \min_{\vec{z}\in \{0,1\}^n} \|\basis \vec{z}-\vec{t}\|^2.
    \end{align*}
    Here, the last equality follows from \cref{cor:euclidean-norm} and the inequality uses the fact that $\vec{c}_{j_1}^{(1)}+\cdots + \vec{c}_{j_k}^{(k)}= \basis^{(1)}\vec{z}_{1}+\cdots +\basis^{(k)}\vec{z}_{k}$ for some $\vec{z}_{1},\cdots, \vec{z}_{k}\in \{0,1\}^{n/k}$. Therefore, we get that if the $(0,1)$-$\CVP$ instance is a \textsf{NO} instance i.e. $\min_{\vec{z}\in \{0,1\}^n} \|\basis \vec{z}-\vec{t}\|> d$ then the reduced minimum weight-k-Clique instance is also a \textsf{NO} instance i.e. every $k$-clique in the graph $\mathcal{G}$ has weight more than $d^2$.

    Now, let's assume that the given $(0,1)$-$\CVP_2$ instance is a \textsf{YES} instance i.e. there exists a vector  $\vec{z}\in \{0,1\}^n$ such that $\|\basis \vec{z}-\vec{t}\|\leq d$. Let $\vec{c}_{j_i}^{(i)}=\vec{B}^{(i)}[z[{(i-1)\frac{n}{k}+1}], \ldots, z[{i\frac{n}{k}}]]^T$. Notice that $\basis \vec{z}=\sum_{i\in [k]}\vec{c}_{j_i}^{(i)}$. We claim that the $k$-Clique formed by the vertices ${v}_{j_1}^{(1)},\ldots, v_{j_k}^{(k)}$ has weight at most $d^2$. The weight of the $k$-clique is 
    \begin{align*}
    \sum_{\substack{i_1, i_2\in [k]^2\\ \text{and } i_1<i_2}} &w(e_{{i_1},j_{i_1}{i_2},j_{i_2}})\\
     =& \sum_{\substack{i_1, i_2\in [k]^2\\ \text{and } i_1<i_2}} 2\langle\vec{c}_{j_{i_1}}^{(i_1)},\vec{c}_{j_{i_2}}^{(i_2)}\rangle+\frac{1}{k-1}\left(\|\vec{c}_{j_{i_1}}^{(i_1)}\|^2+\|\vec{c}_{j_{i_2}}^{(i_2)}\|^2 - 2\langle \vec{c}_{j_{i_1}}^{(i_1)},\vec{t} \rangle-2\langle \vec{c}_{j_{i_2}}^{(i_2)},\vec{t}\rangle\right)+\left(\binom{k}{2}\right)^{-1}\|\vec{t}\|^2 \\
     =& \|\vec{t}\|^2+ \sum_{(i_1, i_2)\xleftarrow{2} [k]}2\langle\vec{c}_{j_{i_1}}^{(i_1)},\vec{c}_{j_{i_2}}^{(i_2)}\rangle+\frac{1}{k-1}\left(\|\vec{c}_{j_{i_1}}^{(i_1)}\|^2+\|\vec{c}_{j_{i_2}}^{(i_2)}\|^2 - 2\langle \vec{c}_{j_{i_1}}^{(i_1)},\vec{t} \rangle-2\langle \vec{c}_{j_{i_2}}^{(i_2)},\vec{t}\rangle\right)\\
     =& \|\vec{t}\|^2+ \sum_{i\in [k]} \left(\|\vec{c}_{j_{i}}^{(i)}\|^2 - 2\langle \vec{c}_{j_{i}}^{(i)},\vec{t} \rangle\right) + \sum_{\substack{i_1, i_2\in [k]^2\\ \text{and } i_1<i_2}}2\langle\vec{c}_{j_{i_1}}^{(i_1)},\vec{c}_{j_{i_2}}^{(i_2)}\rangle\\ 
    =& \|\vec{c}_{j_1}^{(1)}+\cdots+ \vec{c}_{j_k}^{(k)}-\vec{t}\|^2 = \|\basis \vec{z}-\vec{t}\|^2\leq d^2.
    \end{align*}
    This completes the proof.   
\end{proof}

\begin{remark}
    Our proof can be easily extend to reduces a $\gamma$-approximation of \emph{search} $(0,1)$-$\CVP$ to a $\gamma^2$-approximation of \emph{search} minimum-weight-$k$-Clique. 
\end{remark}

\subsection{Algorithms for (0,1)-CVP under the Euclidean norm}

We will use the following result about the algorithm for minimum-weight-triangle.

\begin{theorem}\cite[Theorem~3]{rwilliams08}\label{thm:algo-clique}
    For any positive integers $n,W$, there exists an $\tilde{\mathcal{O}}(Wn^{\omega})$ time algorithm for minimum-weight-triangle on graphs with $n$ vertices where edge weights are from $[-W,W]$
\end{theorem}

We are now ready to prove our main algorithmic result, \cref{thm:main} from the introduction.

\begin{theorem}
    There is an algorithm for $(0,1)$-$\CVP$ that runs in time $ 2^{\omega n/3 +o(n)} \leq  \tilde{\mathcal{O}}((1.7299)^n)$ if the coordinates of the basis and target vectors are bounded by $2^{o(n)}$.
\end{theorem}

\begin{proof}
Suppose we are given a $(0,1)$-$\CVP$ instance with basis $\basis=(\vec{b}_1,\ldots \vec{b}_n)\in \Z^{m\times n} $ and target $\vec{t}\in \Z^m$ where each coordinate is bounded by $2^{\eta}$. By applying \cref{thm:CVPtoclique}, we get an instance of minimum-weight-triangle on a graph with $N=3\cdot 2^{n/3}$ vertices and edge weights bounded by $O(m\cdot 2^{2\eta})$. Then, from \cref{thm:algo-clique}, we get an $\tilde{\mathcal{O}}(m\cdot 2^{2\eta} \cdot (3\cdot 2^{n/3})^{\omega})=\tilde{\mathcal{O}}(m\cdot 2^{2\eta}\cdot 2^{\omega n/3})$ time algorithm. Let $\eta=o(n)$ and for every $(0,1)$-$\CVP$ instance $m$ is considered to be polynomial in $n$. Hence, we get an $\tilde{\mathcal{O}}\left(2^{\omega n/3+o(n)}\right) \leq \tilde{\mathcal{O}}((1.7299)^n)$ time algorithm for $(0,1)$-$\CVP$ if the coordinates of the basis vectors and of the target are bounded by $2^{o(n)}$. 
\end{proof}

By using the reduction from $(0,1)$-$\SVP$ to $(0,1)$-$\CVP$, we also get a $ 2^{\omega n/3 +o(n)} \leq  \tilde{\mathcal{O}}((1.7299)^n)$ time algorithm for $(0,1)$-$\SVP$ if the coordinates of the basis vectors are bounded by $2^{o(n)}$.

\section{From CVP to SAT, and the Equivalence}

In this section, we present a Karp reduction from $(0,1)$-$\CVP_p$ to Weighted Max-$p$-SAT for all $p\in 2\Z_{>0}$, proving ~\cref{thm:cvp-to-sat}.

\ThmCVPtoSAT*
\begin{proof}
Given a $(0,1)$-$\CVP_p$ instance with basis $\basis=(\vec{b}_1,\ldots, \vec{b}_n)\in \Z^{m\times n}$, target $\vec{t}\in \Z^{m}$ and a number $d>0$, we construct a Max-$p$-SAT instance on $n$ variables. 

Let $\mathcal 
{X}=\{x_1,\ldots,x_n\}$ be the set of variables in the formula we construct. For ease of the description, we will also use one more variable $x_{n+1}$ and set $x_{n+1}=1$. In this reduction we will also use $\vec{b}_{n+1}=\vec{t}$. For any tuple of $p$ elements $\mathcal{S}\subset\{X\cup x_{n+1}\}^p$, let $\sigma(\mathcal{S})$ denote the number of occurrences of $x_{n+1}$ in  $\mathcal{S}$ and $\delta(S)$ denote the number of distinct elements in $S$. We create the formula $\Psi$ consisting of the following weighted clauses:
\begin{align*}
  \forall (x_{i_1},\ldots,x_{i_p}) 
  \in \{\mathcal{X} \cup x_{n+1}\}^p, \; &\text{let } k=\delta(x_{i_1},x_{i_2},\ldots,x_{i_p}) \text{ and }x_{i'_1},\ldots, x_{i'_k} \text{ be the set of} \\
  &\text{all  distinct elements from }\{x_{i_1},x_{i_2},\ldots, x_{i_p}\} \\
  &C^0_{i_1,
  \ldots,i_p}:= x_{i'_1} \lor x_{i'_2} \cdots \lor x_{i'_k},\\
  \;&C^1_{i_1,
  \ldots,i_p}:= \overline{x_{i'_1}} \lor x_{i'_2} \cdots \lor x_{i'_k},\\
  \;&C^2_{i_1,
  \ldots,i_p}:= {x_{i'_1}} \lor \overline{x_{i'_2}} \cdots \lor x_{i'_k},\\
  \;&C^3_{i_1,
  \ldots,i_p}:= \overline{x_{i'_1}} \lor \overline{x_{i'_2}} \cdots \lor x_{i'_k},\\
  &\hspace{2em}\vdots\\
  \;&C^{2^k-1}_{i_1,
  \ldots,i_p}:= \overline{x_{i'_1}} \lor \overline{x_{i'_2}} \cdots \lor \overline{x_{i'_k}},\\
  \text{ with weight  }   \forall 0\leq j\leq (2^{k}-2),\;\; &w(C^{j}_{i_1,\ldots,i_p}):=\frac{D}{2^k-1}-\frac{1}{2^{k}-1}\cdot (-1)^{\sigma(\{x_{i_1},\ldots,x_{i_p}\})} \textsf{mvp}(\vec{b}_{i_1},\dots,\vec{b}_{i_p})\\
  \text{ and } w(C^{2^k-1}_{i_1,\ldots,i_p})&:=\frac{D}{2^k-1}+\left(\frac{2^{k}-2}{2^{k}-1}\right)\cdot (-1)^{\sigma(\{x_{i_1},\ldots,x_{i_p}\})} \cdot \textsf{mvp}(\vec{b}_{i_1},\dots,\vec{b}_{i_p})
\end{align*}

It is easy to see that the reduction takes $O(n^p\cdot \poly(n,m,p))$ time and space. Now we will prove the correctness of the reduction.  Observe that for any $(x_{i_1},\ldots,x_{i_p})\in \{\mathcal{X}\cup x_{n+1}\}^p$, any assignment to $\{\mathcal{X} \cup x_{n+1}\}$ variables will satisfy exactly $2^{k}-1$ clauses from $\{C^0_{i_1,
  \ldots,i_p}, C^1_{i_1,
  \ldots,i_p}, \ldots, C^{2^{k}-1}_{i_1,
  \ldots,i_p} \}$. Notice that, 
 \text{if }  $x_{i'_1}=x_{i'_2}=\cdots =x_{i'_k}=1$, \text{then the total weight of satisfied clauses from} $\{C^0_{i_1,
  \ldots,i_p}, C^1_{i_1,
  \ldots,i_p}, \ldots, C^{2^{k}-1}_{i_1,
  \ldots,i_p} \}$ is $D- (-1)^{\sigma(\{x_{i_1},\ldots,x_{i_p}\})} \cdot \textsf{mvp}(\vec{b}_{i_1},\dots,\vec{b}_{i_p})$  \text{and for the rest of the assignments the total weight is} $D$.
  Our correctness proof essentially relies on this observation.

First, we will show that if the given $(0,1)$-$\CVP_p$ instance is a \textsf{YES} instance, $\min_{\vec{z}\in \{0,1\}^n} \|\basis \vec{z}-\vec{t}\|_p\leq d$ , then there exists an assignment to the variables $\mathcal{X}$ that satisfies clauses of $\Psi$ of weight at least $(n+1)^p \cdot D-d^p$ . Let $\vec{z}\in \{0,1\}^n$ be a vector achieving $\|\basis \vec{z}-\vec{t}\|_p\leq d$.  Let $\rho$ be an assignment to the variables in $\mathcal{X}$ where $\forall i \in [n], \rho(x_i) = z_i$, and recall that we always set $\rho(x_{n+1}) = 1$. For any clause $C$, we will use the same notation $\rho(C)$ to indicate whether it satisfies the clause or not, i.e., $\rho(C) = 1$ if the clause $C$ is satisfied by the assignment $\rho$, and $\rho(C) = 0$ otherwise. The total weight of clauses satisfied by the assignment $\rho$ is  
\begin{align*}
\sum_{\substack{(i_1,\dots, i_p)\\ \in [n+1]^p}} &\sum_{\substack{j=0\\ \delta=\delta(x_{i_1},\ldots,x_{i_p})}}^{2^\delta-1} \rho(C^j_{i_1,\ldots,i_p})  w(C^j_{i_1,\ldots,i_p})\\ =&\sum_{\substack{(i_1,\dots, i_p) \in [n+1]^p\\ \text{ and } \rho(x_{i_1})=\cdots \rho(x_{i_p})=1}}  D- (-1)^{\sigma({x_{i_1},\ldots,x_{i_p}})}\mathsf{mvp}(\vec{b}_{i_1},\ldots, \vec{b}_{i_p})
    +\sum_{\substack{(i_1,\dots, i_p) \in [n+1]^p\\ \text{ and } \exists l\in [p] \text{ s.t. } \rho(x_{i_l})\neq 1}}  D\\
=&(n+1)^p\cdot  D -\sum_{\substack{(i_1,\dots, i_p) \in [n+1]^p\\ \text{ and } z_{i_1}=\cdots z_{i_p}=1}}(-1)^{\sigma({x_{i_1},\ldots,x_{i_p}})}\mathsf{mvp}(\vec{b}_{i_1},\ldots, \vec{b}_{i_p})\\
=&(n+1)^p\cdot  D-\|\basis \vec{z}-\vec{b}_{n+1}\|_p^p =
(n+1)^p\cdot  D-\|\basis \vec{z}-\vec{t}\|_p^p
\geq (n+1)^p\cdot  D-d^p.
\end{align*}

Here the first equality follows from the observation mentioned in  above paragraph and the third equality follows from the \cref{cor:simple-mvp}.

Now, we will show that if the given $(0,1)$-$\CVP_p$ instance is a \textsf{NO} instance, $\min_{\vec{z}\in \{0,1\}^n} \|\basis \vec{z}-\vec{t}\|_p>d$, then all assignments to the variables $\mathcal{X}$ satisfy clauses of $\Psi$ of weight less than $(n+1)^p\cdot D-d^p$. For the sake of contradiction, let's assume that an assignment $\rho$ satisfies clauses of weight greater than equal to $(n+1)^p\cdot (2^p-1)D-d^p$. Recall that we have already fixed $\rho({x_{n+1}})=1$.

\begin{align*}
\sum_{\substack{(i_1,\dots, i_p)\\ \in [n+1]^p}} \sum_{\substack{j=0\\ \delta=\delta(x_{i_1},\ldots,x_{i_p})}}^{2^\delta-1} \rho(C^j_{i_1,\ldots,i_p}) & w(C^j_{i_1,\ldots,i_p})\\ =&\sum_{\substack{(i_1,\dots, i_p) \in [n+1]^p\\ \text{ and } \rho(x_{i_1})=\cdots \rho(x_{i_p})=1}}  D- (-1)^{\sigma({x_{i_1},\ldots,x_{i_p}})}\mathsf{mvp}(\vec{b}_{i_1},\ldots, \vec{b}_{i_p})\\
& +\sum_{\substack{(i_1,\dots, i_p) \in [n+1]^p\\ \text{ and } \exists l\in [p] \text{ s.t. } \rho(x_{i_l})\neq 1}}  D\\
=&(n+1)^p\cdot D -\sum_{\substack{(i_1,\dots, i_p) \in [n+1]^p\\ \text{ and } z_{i_1}=\cdots z_{i_p}=1}}(-1)^{\sigma({x_{i_1},\ldots,x_{i_p}})}\mathsf{mvp}(\vec{b}_{i_1},\ldots, \vec{b}_{i_p})\\
=&(n+1)^p\cdot  D-\left\|\sum_{i=1}^n \rho(x_i)\vec{b_i}-\rho(x_{n+1})\vec{b}_{n+1}\right\|_p^p\\
=&(n+1)^p\cdot D-\left\|\sum_{i=1}^n \rho(x_i)\vec{b_i}-\vec{t}\right\|_p^p\\
\geq & (n+1)^p\cdot  D-\left(\min_{\vec{z}\in \{0,1\}^n} \|\basis \vec{z}-\vec{t}\|_p\right)^p\\ \\
>&(n+1)^p\cdot  D-d^p.
\end{align*}
This gives a contradiction. Hence, it also completes the proof. 
\end{proof}

We recall the following result by Aggarwal, Bennett, Golovnev, and Stephens-Davidowitz \cite{ABGS19}.

\begin{theorem}(\cite[Theorem~3.2]{ABGS19} and \cite[Theorem 3.2]{BGS17})\label{thm:sat-to-cvp}
For any positive integer $n$ and even integer $p$, there exists a Karp reduction from Max-$p$-$\SAT$ on $n$ variables to $(0,1)$-$\CVP_p$ on lattices of rank $n$.
\end{theorem}

\remark In \cite{ABGS19}, Aggarwal, Bennett, Golovnev, and Stephens-Davidowitz show constructions of $(p,p)$-isolating parallelepiped. Combining this with Theorem 3.2 of \cite{BGS17} gives the above theorem. The result is more general, but we write a specific part of it which is required in this paper.

The proof of \cref{cor:equiv} in the introduction now follows directly from \cref{thm:cvp-to-sat} and \cref{thm:sat-to-cvp}.

\paragraph{A Karp reduction from SVP to SAT:} In \cite{GMSS99}, Goldreich, Micciancio, Safra, and Seifert presented a polynomial-time reduction from $\SVP_p$ on an $n$-rank lattice to $\CVP_p$ on an $n$-rank lattice. The reduction involves making $n$ calls to $\CVP_p$. By a trivial modification to this reduction, we also get a polynomial time reduction from $(0,1)$-$\SVP_p$ to $(0,1)$-$\CVP_p$. It will also require $n$ calls to $(0,1)$-$\CVP_p$. By combining this result with \cref{thm:cvp-to-sat}, we obtain a Turing reduction from $(0,1)$-$\SVP_p$ to Max-$p$-$\SAT$ that requires $n$ calls to Max-$p$-$\SAT$.

Now, the question arises: ``Can we obtain a  Karp reduction from $\SVP_p$ to $\SAT$?'' Using a similar idea as the previous reduction, we can achieve a Karp reduction from $(0,1)$-$\SVP_p$ to Max-$p$-$\SAT$ that has a factor $2$ blowup. Specifically, for any $p\in 2\Z_{>0}$, we can devise a polynomial-time reduction from $(0,1)$-$\SVP_p$ on an $n$-rank lattice to Max-$p$-$\SAT$ on $2n$ variables. The main challenge in this reduction is to ensure that the Max-$p$-$\SAT$ solver produces a non-zero solution. This can be accomplished by introducing a single clause $x_1\lor x_2\lor \cdots \lor x_n$ with a high weight. Additionally, we will need $n$ additional variables to convert this clause into $k$-$\SAT$ formulas.

\section{From CVP in even norms to $k$-Clique on Hypergraphs}
In this section, we present a Karp-reduction from $(0,1)$-$\CVP_p$ to minimum-weight-$k$-Clique on $p$-uniform hypergraph. It is an extension of the reduction given in \cref{sec:cvp-to-clique} to any $\ell_p$ norm for any even integer $p\geq 2$.

We will use the following corollary of \cref{lem:mvp}. 

\begin{restatable}{corollary}{CorEvenNorm}\label{cor:even-norm}
	For any $p\in 2\Z_{>0}$ and vectors $\vec{v}_1,\ldots, \vec{v}_k,\vec{v}_{k+1}$, we have 
	\[\|\vec{v}_1+\cdots+\vec{v}_k-\vec{v}_{k+1}\|_p^p=\sum_{\substack{(i_1,\ldots,i_p)\in[k]^p\\ \text{and }i_1<i_2\ldots <i_p  }}\sum\limits_{ \substack{\mathcal{X}=(l_1,\ldots,l_p) \\ \in (i_1,\ldots,i_p, k+1\}^p} } (-1)^{\sigma(\mathcal{X})}\cdot \frac{1}{\beta(k,p,\mathcal{X})}
\cdot  \textbf{mvp}\left(\vec{v}_{{l_1}},\ldots, \vec{v}_{{l_p}}\right),\]
	where $\sigma(\mathcal{X})$ denotes the number of occurrences of $k+1$ in the tuple $\mathcal{X}\in [k+1]^p $, $\beta(k,p,\mathcal{X})=\binom{k-\sigma'(\mathcal{X})}{p-\sigma'(\mathcal{X})}$, and $\sigma'(\mathcal{X})$ denotes the number of distinct elements except $k+1$ in the tuple $\mathcal{X}\in [k+1]^p$.
\end{restatable}
 We defer the proof to \cref{appendix:cor-even-norm}.

\begin{theorem}\label{thm:CVP-even-norm-to-clique}
  For any  positive integers $n,k\geq 2$, and $p\in 2\Z_{>0}$ satisfying $p\leq k\leq n$, there is a Karp-reduction from $\{0,1\}$-$\CVP_p$ on lattices of rank $n$ to minimum-weight-$k$-Clique on $p$-uniform hypergraphs on $N=k\cdot 2^{\lceil n/k\rceil}$ vertices \\ Furthermore, the reduction takes only $O(\poly(n,m) \cdot (k^2\cdot 2^{n/k})^p)$ time and space.
\end{theorem}

\begin{proof}
 Given a $(0,1)$-$\CVP_p$ instance with basis $\basis=(\vec{b}_1,\ldots \vec{b}_n)\in \Z^{m\times n} $, a target $\vec{t}\in \Z^m$ and a number $d>0$, we construct a minimum-weight-$k$-Clique instance on a $p$-uniform hypergraph $\mathcal{G}=(\mathcal{V},\mathcal{E})$ on $k\cdot 2^{n/k}$ vertices. For simplicity, let's assume that $n$ is an integer multiple of $k$.

    First, partition the basis vectors into $k$ sets: $\basis^{(1)},\ldots, \basis^{(k)}$, each consisting of $n/k$ vectors such that $\basis^{(i)}=\left\{\vec{b}_{\frac{n}{k}(i-1)+1},\ldots, \vec{b}_{i\frac{n}{k}}\right\}$. We construct $\mathbf{C}^{(i)}=[\vec{c}_{1}^{(i)}\ldots \vec{c}_{2^{n/k}}^{(i)} ]\in \Z^{m\times 2^{n/k}}$ for $i\in [k]$, where the columns of $\mathbf{C}^{(i)}$ represent $\{0,1\}$-combinations of vectors from $\mathbf{B}^{(i)}$ and each vector in $\mathbf{C}^{(i)}$ will correspond to a vertex in the reduced hypergraph.  Let
     $\mathcal{V}:=\{v_{j}^{(i)}| i\in [k], j\in [2^{n/k}]\}$ 
     be the set of vertices and 
     \[\mathcal{E}:=\{e_{i_1,j_1,\ldots, i_p,j_p}=(v_{j_1}^{(i_1)},\ldots,v_{j_p}^{(i_p)})| i_1,\ldots, i_p\in [k]^p, i_1<\ldots < i_p,  j_1,\ldots,j_{p}\in [2^{n/k}] \}\]
     be the edge set. Notice that it is a $k$-partite hypergraph. Let $\vec{c}_{1}^{(k+1)}=\vec{t}$ and for any tuple of $p$ pairs of elements $\mathcal{S}\subset (\Z, \Z)^p$, let $\sigma(\mathcal{S})$ denote the number of pairs equal to $(k+1,1)$, and $\sigma'(\mathcal{S})$ denote the number of distinct pairs other than $(k+1,1)$ in the tuple $\mathcal{S}$.  For an edge $e_{i_1,j_1,\ldots, i_p,j_p}\in \mathcal{E}$, we define the edge weight as follows: let  $\mathcal{S}=\left((i_1,j_1),\ldots,(i_p,j_p), (k+1,1)\right)$
    
       \[ w(e_{i_1,j_1,\ldots,i_p,j_p})= \sum\limits_{\mathcal{X}=((i_{l_1},j_{l_1}),\ldots,(i_{l_p},j_{l_p})) \in \mathcal{S}^p } (-1)^{\sigma(\mathcal{X})}\frac{1}{\beta(k,p,\mathcal{X})}\cdot  \textbf{mvp}\left(\vec{c}_{j_{l_1}}^{(i_{l_1})},\ldots, \vec{c}_{j_{l_p}}^{(i_{l_p})}\right),\]
   where $\beta(k,p,\mathcal{X}):=\binom{k-\sigma'(\mathcal{X})}{p-\sigma'(\mathcal{X})}$.

    It is trivial that the reduction takes $O(\poly(n,m) \cdot (k^2\cdot 2^{n/k})^p)$ time and space. In the rest of the proof, we will show that there exists a $k$-clique of weight less than equal to $d^p$ if the given $(0,1)$-$\CVP_p$ instance is a \textsf{YES} instance. Otherwise, all $k$-cliques have weight greater than $d^p$.
    
    We claim that the weight of a $k$-clique on the reduced graph is at least $\min_{\vec{z}\in \{0,1\}^n} \|\basis \vec{z}-\vec{t}\|_p^p$. Notice that, any $k$-clique in a $k$-partite hypergraph must have exactly one vertex from each partition.  Let the vertices ${v}_{j_1}^{(1)},\ldots, {v}_{j_k}^{(k)}$ form a $k$-clique, then the weight of the clique is 
    \begin{align*}
    &\sum_{\substack{\{x_1,\ldots, x_p\}\subset [k]}} w(e_{{x_1},j_{x_1}},\ldots, e_{{x_p},j_{x_p}})\\
    =& \sum_{\substack{\{x_1,\ldots, x_p\}\subset [k]}}\sum\limits_{ \substack{\mathcal{X}=(({l_1},j_{l_1}),\ldots,({l_p},j_{l_p})) \\ \in \{({x_1},j_{x_1}),\ldots,({x_p},j_{x_p})\}^p} } (-1)^{\sigma(\mathcal{X})}\cdot \frac{1}{\beta(k,p,\mathcal{X})}\cdot  \textbf{mvp}\left(\vec{c}_{j_{l_1}}^{({l_1})},\ldots, \vec{c}_{j_{l_p}}^{({l_p})}\right)\\
    =& \|\vec{c}_{j_1}^{(1)}+\cdots + \vec{c}_{j_k}^{(k)}-\vec{t}\|_p^p \geq \min_{\vec{z}\in \{0,1\}^n} \|\basis \vec{z}-\vec{t}\|_p^p.
    \end{align*}
    Here, the second equality follows from \cref{cor:even-norm} and the inequality uses the fact that $\vec{c}_{j_1}^{(1)}+\cdots + \vec{c}_{j_k}^{(k)}= \basis^{(1)}\vec{z}_{1}+\cdots +\basis^{(k)}\vec{z}_{k}$ for some $\vec{z}_{1},\cdots, \vec{z}_{k}\in \{0,1\}^{n/k}$.  Therefore, we get that if the $(0,1)$-$\CVP_p$ instance is a \textsf{NO} instance i.e. $\min_{\vec{z}\in \{0,1\}^n} \|\basis \vec{z}-\vec{t}\|_p>d$ then the reduced minimum weight-k-Clique instance is a \textsf{NO} instance i.e. every $k$-clique in the hypergraph $\mathcal{G}$ has weight more than $d^p$.

    Now, let's assume that the given $(0,1)$-$\CVP_p$ instance is a \textsf{YES} instance i.e. there exists $\vec{z}\in \{0,1\}^n$ such that $\|\basis \vec{z}-\vec{t}\|_p\leq d$. Let $\vec{c}_{j_i}^{(i)}=\vec{B}^{(i)}\left[z[{(i-1)\frac{n}{k}+1}], \ldots, z[i\frac{n}{k}]\right]^T$. Notice that $\basis \vec{z}=\sum_{i\in [k]}\vec{c}_{j_i}^{(i)}$. We claim that the $k$-clique formed by the vertices ${v}_{j_1}^{(1)},\ldots, v_{j_k}^{(k)}$ has weight less than $d^p$. The weight of the $k$-clique is 
     \begin{align*}
    \sum_{\substack{\{i_1,\ldots, i_p\}\subset [k]}} &w(e_{{i_1},j_{i_1}},\ldots, e_{{i_p},j_{i_p}})\\
     =& \sum_{\substack{\{i_1,\ldots, i_p\}\subset [k]}}\sum\limits_{ \substack{\mathcal{X}=(({l_1},j_{l_1}),\ldots,({l_p},j_{l_p})) \\ \in \{(i_1,j_{i_1}),\ldots,(i_p,j_{i_p})\}^p}}  (-1)^{\sigma(\mathcal{X})}\cdot \frac{1}{\beta(k,p,\mathcal{X})}\cdot  \textbf{mvp}\left(\vec{c}_{j_{l_1}}^{({l_1})},\ldots, \vec{c}_{j_{l_p}}^{(l_p)}\right)\\
    =& \|\vec{c}_{j_1}^{(1)}+\cdots+ \vec{c}_{j_k}^{(k)}-\vec{t}\|_p^p = \|\basis \vec{z}-\vec{t}\|_p^p\leq d^p.
    \end{align*}
    Here, the second equality follows from \cref{cor:even-norm} and  the inequality uses the condition $\basis \vec{z}=\sum_{i\in [k]}\vec{c}_{j_i}^{(i)}$. This completes the proof.   
\end{proof}

\bibliographystyle{alpha}
\bibliography{ref}
\appendix

\section{Proof of \cref{cor:even-norm}}\label{appendix:cor-even-norm}
\CorEvenNorm*
\begin{proof}

	From $\cref{lem:mvp}$, we have
	\begin{align*} &\|\vec{v}_1+\cdots+\vec{v}_k-\vec{v}_{k+1}\|_p^p=\sum_{\mathcal{X}=(i_1,\ldots,i_p)\in [k+1]^p} (-1)^{\sigma(\mathcal{X})}\mvp(\vec{v}_{i_1},\ldots, \vec{v}_{i_p}) \\ 
		=&\sum_{\substack{\mathcal{X}=(i_1,\ldots,i_p)\in [k+1]^p\\ \text{ and } \sigma'(\mathcal{X})=p}} (-1)^{\sigma(\mathcal{X})} \mvp(\vec{v}_{i_1},\ldots, \vec{v}_{i_p})+ \sum_{\substack{\mathcal{X}=(i_1,\ldots,i_p)\in [k+1]^p\\ \text{ and } \sigma'(\mathcal{X})=p-1}}
		(-1)^{\sigma(\mathcal{X})} \mvp(\vec{v}_{i_1},\ldots, \vec{v}_{i_p})+\\
		 &\cdots+ \sum_{\substack{\mathcal{X}=(i_1,\ldots,i_p)\in [k+1]^p\\ \text{ and } \sigma'(\mathcal{X})=0}}
		(-1)^{\sigma(\mathcal{X})} \mvp(\vec{v}_{i_1},\ldots, \vec{v}_{i_p}) \\
=&\sum_{\substack{\mathcal{Y}=\{i_1,\ldots,i_p\}\subset [k]}} \sum_{\substack{\mathcal{X}=(i_{l_1},\ldots,i_{l_p}) \\\in \{\mathcal{Y}\cup (k+1)\}^p\text{ and } \sigma'(\mathcal{X})=p}} (-1)^{\sigma(\mathcal{X})} \mvp(\vec{v}_{i_1},\ldots, \vec{v}_{i_p})+ \\ &\sum_{\substack{\mathcal{X}=(i_1,\ldots,i_p)\in [k+1]^p\\ \text{ and } \sigma'(\mathcal{X})=p-1}}
	(-1)^{\sigma(\mathcal{X})}	\mvp(\vec{v}_{i_1},\ldots, \vec{v}_{i_p})+\cdots+ \sum_{\substack{\mathcal{X}=(i_1,\ldots,i_p)\in [k+1]^p\\ \text{ and } \sigma'(\mathcal{X})=0}}
		(-1)^{\sigma(\mathcal{X})} \mvp(\vec{v}_{i_1},\ldots, \vec{v}_{i_p}) \\
		=& \sum_{\substack{\mathcal{Y}=\{i_1,\ldots,i_p\}\subset [k] }}\Biggl( \sum_{\substack{\mathcal{X}=(i_{l_1},\ldots,i_{l_p}) \in \{\mathcal{Y}\cup (k+1)\}^p\\ \text{ and } \sigma'(\mathcal{X})=p}}(-1)^{\sigma(\mathcal{X})}\mvp(\vec{v}_{i_1},\ldots, \vec{v}_{i_p})+ \\		
  & \sum_{\substack{\mathcal{X}=(i_{l_1},\ldots,i_{l_p}) \in \{\mathcal{Y}\cup (k+1)\}^p\\ \sigma'(\mathcal{X})=p-1}} (-1)^{\sigma(\mathcal{X})} \frac{1}{k-(p-1)}  \mvp(\vec{v}_{i_1},\ldots, \vec{v}_{i_p})\Biggl) 
		+\cdots + \\ &\sum_{\substack{\mathcal{X}=(i_1,\ldots,i_p)\in [k+1]^p\\ \text{ and } \sigma'(\mathcal{X})=0}}
		(-1)^{\sigma(\mathcal{X})} \mvp(\vec{v}_{i_1},\ldots, \vec{v}_{i_p}) \\
=&\sum_{\substack{\{i_1,\ldots,i_p\}\subset [k]  }} \left(\sum\limits_{ \substack{\mathcal{X}=(l_1,\ldots,l_p) \\ \in \{x_1,\ldots,x_p, k+1\}^p} } (-1)^{\sigma(\mathcal{X})}\cdot \binom{k-\sigma'(\mathcal{X})}{p-\sigma'(\mathcal{X})}^{-1}\cdot  \textbf{mvp}\left(\vec{v}_{{l_1}},\ldots, \vec{v}_{{l_p}}\right)\right)\\
 =&\sum_{\substack{\{i_1,\ldots,i_p\}\subset [k]  }} \left(\sum\limits_{ \substack{\mathcal{X}=(l_1,\ldots,l_p) \\ \in \{i_1,\ldots,i_p, k+1\}^p} } (-1)^{\sigma(\mathcal{X})}\cdot \frac{1}{\beta(k,p,\mathcal{X})}
\cdot  \textbf{mvp}\left(\vec{v}_{{l_1}},\ldots, \vec{v}_{{l_p}}\right)\right). 
\end{align*}  
\end{proof}

\end{document}